\documentclass[twocolumn,amsmath,amssymb,10pt,aps,tikz,border=5mm]{revtex4}
  
\usepackage[utf8]{inputenc}
\usepackage{amsmath,braket}
\usepackage{amsfonts}
\usepackage{graphicx}
\usepackage{color}
\usepackage{amsthm}
\usepackage{bm}
\usepackage{bbm}
\usepackage{mathtools}
\usepackage{hyperref}
\hypersetup{
    colorlinks=true,
    linkcolor=blue,
    filecolor=magenta,      
    urlcolor=cyan,
}

\newtheorem{thm}{Theorem}

\usepackage{placeins}
\renewcommand{\eqref}[1]{Eq.~\ref{#1}}

\newtheorem{Definition}{Definition}

\begin{document}

\title{\bf Dynamic state reconstruction of quantum systems subject to pure decoherence}
\author{Artur Czerwinski}
\email{aczerwin@umk.pl}
\affiliation{Institute of Physics, Faculty of Physics, Astronomy and Informatics \\ Nicolaus Copernicus University,
Grudziadzka 5, 87--100 Torun, Poland}

\begin{abstract}
The article introduces efficient quantum state tomography schemes for qutrits and entangled qubits subject to pure decoherence. We implement the dynamic state reconstruction method for open systems sent through phase-damping channels which was proposed in: Open Syst. Inf. Dyn. 23, 1650019 (2016). In the current article we prove that two distinct observables measured at four different time instants suffice to reconstruct the initial density matrix of a qutrit with evolution given by a phase-damping channel. Furthermore, we generalize the approach in order to determine the optimal criteria for quantum tomography of entangled qubits. Finally, we prove two universal theorems concerning the minimal number of distinct observables required for quantum tomography of qudits. We believe that dynamic state reconstruction schemes bring significant advancement and novelty to quantum tomography since they allow to reduce the number of distinct measurements required to solve the problem, which is important from the experimental point of view.
\end{abstract}
\maketitle

\section{Introduction}
Quantum state tomography, which means the problem of reconstructing the accurate representation of a physical system from measurements, is crucial for quantum information and computation. It has been an important line of research since 1933 when Wolfgang Pauli asked whether a wave function can be determined from the position and momentum distributions. For years much has been written about quantum tomography (e.g. \cite{james01,dariano03,paris04,wasilewski07,kolenderski09,banaszek13}) but still many fundamental questions remain to be answered.

The ability to recover quantum states from measurements is especially relevant to quantum communication which requires well-characterized quantum resources to encode information. Quantum tomography schemes are commonly used in order to evaluate the effectiveness of information processing in case of imperfect measurements scenario \cite{molina04,rosset12}.

Among many different methods of quantum tomography special attention should be paid to the so-called stroboscopic approach which was originally introduced in 1983 by Andrzej Jamiolkowski \cite{jam83}. This approach aims to reconstruct the initial density matrix by the lowest possible number of distinct observables due to utilization of the knowledge about evolution of quantum systems (encoded, for example, in a GKSL equation\cite{gorini76,marmo19}). For a given linear generator of evolution one can compute the index of cyclicity, which expresses the minimal number of distinct observables required for quantum tomography \cite{jam04,czerwin17}. The mean value of each observable is measured at several time instants over different copies of the system (prepared in the same unknown initial state). The stroboscopic approach to quantum tomography appears to have a great potential for experimental applications since it allows to obtain the unknown density matrix by devising few measurement setups.

In \cite{czerwin16b} the stroboscopic tomography has been reformulated and applied to quantum systems which are sent through phase-damping channels (pure decoherence). This kind of quantum dynamics stems from fundamental results of the theory of open systems and properties of such channels have been studied \cite{helm11}. Therefore, it seems utterly justified to analyze this evolution model in the context of dynamic quantum tomography. In \cite{czerwin16b} quantum tomography problem was solved for two specific phase-damping channels -- qubits subject to dephasing and qudits with evolution given by Gaussian semigroup.

In the section \ref{revision} we revise the theoretical background concerning the dynamic quantum tomography model for systems subject to pure decoherence. Then we introduce novel results which are included in the three subsequent sections. First, we prove that one can reconstruct the initial state of a qutrit on the basis of mean values of $2$ observables measured at $4$ distinct time instants. Next we analyze the optimal criteria for quantum tomography of $4-$level quantum systems. Special attention is paid to entangled qubits as a classic example of such systems. Finally, we formulate general theorems concerning qudits subject to pure decoherence. We prove that two observables are sufficient to determine all off-diagonal elements of the density matrix. In addition we give the upper bound for the number of distinct observables required for quantum tomography of $N-$level open systems subject to pure decoherence.

Throughout the article we shall use the following notations. Let $\mathcal{H}$ stand for a finite dimensional Hilbert space ($dim\mathcal{H} = N < \infty$) associated with the physical system. Then by $\mathcal{S(H)}$ we shall denote the state set, i.e. the set of all legitimate density matrices: $\mathcal{S(H)} = \{\rho: \mathcal{H} \rightarrow \mathcal{H},\text{ } \rho \geq 0,\text{ } Tr \rho = 1\}$. Furthermore, we shall use $\mathcal{B(H)}$ in reference to the space of all linear operators on $\mathcal{H}$ and $\mathcal{B_* (H)}$ to denote the space of all Hermitian (self-adjoint) operators on $\mathcal{H}$. Finally $ \mathbb{M}_N (\mathbb{C})$ shall denote the vector space of $N \times N$ complex matrices.

\section{Quantum tomography for phase-damping channels -- revision}\label{revision}

In this section, we shall briefly revise the dynamic approach to quantum tomography of open systems subject to pure decoherence introduced in \cite{czerwin16b}. First, let us recall the definition of a phase-damping channel, assuming that $dim\mathcal{H}=N$.

\begin{Definition}\label{defn1}
A quantum channel of the following form:
\begin{equation}\label{17}
\rho(t) = D(t) \circ \rho(0),
\end{equation}
where $D(t)$ satisfies the conditions:
\begin{equation}\label{conditions}
\begin{split}
&1.\hspace{10pt} \forall_{t \geq 0} \hspace{10pt} D(t) \geq 0, \\
&2. \hspace{10pt} \forall_{t \geq 0} \hspace{10pt} d_{ii} (t) =1 \hspace{10pt}\text{for}\hspace{10pt}  i=1,\dots, N,\\
&3.\hspace{10pt} d_{ij}(0) =1 \hspace{10pt}\text{for}\hspace{10pt} i,j=1,\dots, N.
\end{split}
\end{equation}
is called a phase-damping channel.
\end{Definition}

The conditions enumerated in the definition guarantee that the quantum channel \eqref{17} is a completely positive and trace-preserving (CPTP) map. The matrix $D(t)$ shall be referred to as the dynamic matrix.

The goal of quantum tomography is to reconstruct the initial density matrix $\rho(0)$ based on data accessible from an experiment. We assume that there is a set of observables $\{Q_1, \dots, Q_r\}$ and each of them can be measured at discrete time instants $\{t_1, \dots, t_p\}$. Thus, from an experiment we can obtain a matrix of data, denoted by $[m_i(t_j)]$. The elements of this matrix can be expressed as:
\begin{equation}\label{18}
m_i (t_j) = Tr\left\{ Q_i (D(t_j) \circ \rho(0))\right\},
\end{equation}
where $i=1,\dots, r$ and $j=1, \dots, p$.

In \cite{czerwin16b} it was demonstrated that any continuous time-dependent matrix $D(t) \in \mathbb{M}_N (\mathbb{C})$ can be decomposed in the basis of $\mu$ linearly independent matrices $A_k \in \mathbb{M}_N (\mathbb{C})$:
\begin{equation}\label{decom}
D(t) = \sum_{k=1}^{\mu} \alpha_k (t) A_k,
\end{equation}
where $\{\alpha_k (t)\}$ denotes a set of linearly independent functions $\alpha_k (t): \mathbb{R} \rightarrow \mathbb{C}$.

By implementing the relation between the Hadamard product and the standard matrix product (c.f. \cite{schott05,czerwin16b}), one can write a very useful formula for the measurement results:
\begin{equation}\label{21}
m_i (t_j) = \sum_{k=1}^{\mu} \alpha_k (t_j) Tr\left\{ (Q_i \circ A_k^T) \rho(0) \right\}.
\end{equation}

One can notice that if the measurement of any observable $Q_i$ is performed at distinct time instants $t_1,\dots, t_p$ we obtain a set of $p$ equations:

\begin{equation}\label{22}
\begin{split}
& m_i (t_1) = \sum_{k=1}^{\mu} \alpha_k (t_1) Tr\left\{ (Q_i \circ A_k^T) \rho(0) \right\}, \\
&m_i (t_2) = \sum_{k=1}^{\mu} \alpha_k (t_2) Tr\left\{ (Q_i \circ A_k^T) \rho(0) \right\}, \\
& \vdots \\
& m_i (t_p) = \sum_{k=1}^{\mu} \alpha_k (t_p) Tr\left\{ (Q_i \circ A_k^T) \rho(0) \right\}. 
\end{split}
\end{equation}

On the left-hand side of the system \eqref{22} one has a vector of data which is obtainable from an experiment, whereas on the right-hand side there is a matrix $[\alpha_k (t_j)]$ which is computable based on the algebraic structure of $D(t)$. Thus, one can agree that the figures $Tr\left\{ (Q_i \circ A_k^T) \rho(0) \right\}$ (where $k=1,\dots,\mu$) are computable from \eqref{22} if and only if $rank [\alpha_k (t_j)] = \mu$. This condition claims that the matrix $[\alpha_k (t_j)]$ has to have full rank, which, in practice, means that the number of measurements equals $\mu$ and $[\alpha_k (t_j)]$ is a square matrix.

Due to repeated measurements of the same observable (over distinct copies of the system) we compute the projections of $\rho(0)$ onto a set of operators: $\{Q_i \circ A_1^T, \dots, Q_i \circ A_{\mu}^T \}$. The measurement procedure is then performed for each observable from the set $\{Q_1, \dots, Q_r\}$.

In order to be able to reconstruct the initial density matrix from the projections $Tr\left\{ (Q_i \circ A_k^T) \rho(0) \right\}$, the set of operators $\{Q_i \circ A_k^T\}$ where $ i=1,\dots,r$ and $k=1,\dots,\mu$ have to span $\mathcal{B_* (H)}$, i.e. the space to which $\rho(0)$ belongs. In other words, the set $\{Q_i \circ A_k^T\}$ where $ i=1,\dots,r$ and $k=1,\dots,\mu$ has to be informationally complete.

\section{Dynamic quantum tomography model for a qutrit subject to pure decoherence}\label{qutrit}

\subsection{Preliminaries}

In case of qutrits (i.e. $dim \mathcal{H} =3$), in order to find a basis for the set of density matrices one can refer to the Gell-Mann matrices. We shall follow the notation from \cite{gellmann62} and use $\{\lambda_1,\lambda_2, \dots, \lambda_8\}$ to denote the matrices:

\begin{equation}
\begin{split}
& \lambda_1 = \left [ \begin{matrix} 0 & 1  & 0 \\ 1 & 0 & 0 \\ 0 & 0 &0 \end{matrix} \right], \hspace{0.2cm} \lambda_2 = \left [ \begin{matrix} 0 & -i  & 0 \\ i & 0 & 0 \\ 0 & 0 &0 \end{matrix} \right], \hspace{0.2cm}  \lambda_3 = \left [ \begin{matrix} 1 & 0 & 0 \\ 0 & -1 & 0 \\ 0 & 0 &0 \end{matrix} \right], \\
& \lambda_4 = \left [ \begin{matrix} 0 & 0  & 1 \\ 0 & 0 & 0 \\ 1 & 0 &0 \end{matrix} \right], \hspace{0.25cm}  \lambda_5 = \left [ \begin{matrix} 0 & 0  & -i \\ 0 & 0 & 0 \\ i & 0 &0 \end{matrix} \right], \hspace{0.25cm}  \lambda_6 = \left [ \begin{matrix} 0 & 0 & 0 \\ 0 & 0 & 1 \\ 0 & 1 &0 \end{matrix} \right], \\
&\hspace{0.75cm}  \lambda_7 = \left [ \begin{matrix} 0 & 0 & 0 \\ 0 & 0 & -i \\ 0 & i &0 \end{matrix} \right], \hspace{0.3cm}  \lambda_8 =  \frac{1}{\sqrt{3}}\left [ \begin{matrix} 1 & 0 & 0 \\ 0 & 1 & 0 \\ 0 & 0 & -2 \end{matrix} \right].\\
\end{split}
\end{equation}

The Gell-Mann matrices are a generalization of the Pauli operators and they span the Lie algebra of the $SU(3)$ group. Thus, they satisfy the conditions:
\begin{equation}
\lambda_i = \lambda_i ^*, \hspace{0.5cm} Tr \lambda_i =0 \hspace{0.25cm} \text{and} \hspace{0.25cm} Tr \lambda_i \lambda_j = 2 \delta_{\ij}.
\end{equation}

For a $3-$level quantum system the unknown density matrix $\rho(0) \in \mathcal{S(H)}$ can be decomposed in the basis of the generators of $SU(3)$\cite{genki03}:
\begin{equation}\label{e3.4}
\rho(0) = \frac{1}{3} \mathbb{I}_3 + \frac{1}{2} \sum_{i=1}^8 \langle \lambda_i \rangle \lambda_i,
\end{equation}
where $\langle \lambda_i \rangle$ is the expectation value of the observable $\lambda_i$ and can mathematically be represented as $\langle \lambda_i \rangle = Tr\{\lambda_i \rho(0)\}$.

If one would like to solve the quantum tomography problem for a qutrit without taking advantage of the knowledge about evolution, it would be necessary to measure $8$ distinct observables: $\{\lambda_1,\lambda_2, \dots, \lambda_8\}$ and then use their mean values to complete the formula for $\rho(0)$. Such a static approach to quantum state tomography appears impractical since in a laboratory it would be impossible to define $8$ distinct physical quantities.

\subsection{Evolution model for qutrits}

Let us consider a quantum system such that its trajectory is given by a time-dependent channel:

\begin{equation}\label{e3.1}
\rho(t) = \left[ \begin{matrix} 1 & e^{-\gamma_1 t} & e^{-\gamma_2 t} \\ e^{-\gamma_1 t} & 1 & e^{-\gamma_3 t} \\ e^{-\gamma_2 t} & e^{-\gamma_3 t} & 1 \end{matrix} \right] \circ \rho(0),
\end{equation}

where $\gamma_1, \gamma_2, \gamma_3$ are positive decoherence rates (we assume that $\gamma_1 \neq \gamma_2 \neq \gamma_3$) and $\rho(0)$ denotes the unknown initial density matrix of a $3-$level system ($\rho(0) \in \mathcal{S(H)}$).

One can quickly check that the dynamic matrix from \eqref{e3.1}, i.e.:
\begin{equation}\label{e3.3}
D(t) =  \left[ \begin{matrix} 1 & e^{-\gamma_1 t} & e^{-\gamma_2 t} \\ e^{-\gamma_1 t} & 1 & e^{-\gamma_3 t} \\ e^{-\gamma_2 t} & e^{-\gamma_3 t} & 1 \end{matrix} \right],
\end{equation}
satisfies the conditions  \eqref{conditions}. Therefore, the time-dependent channel from \eqref{e3.1} describes a legitimate evolution, i.e. $\rho(t) \in \mathcal{S(H)}$ holds for all $t\geq 0$. The quantum channel defined in \eqref{e3.1} belongs to the family of phase-damping channels, which means that we can implement the method from \cite{czerwin16b} in order to reconstruct the initial density matrix.

\subsection{Results and analysis}

The main goal of this subsection is to prove that by taking advantage of the knowledge about the evolution (which is given by \eqref{e3.1}) one can significantly decrease the number of distinct observables required for the ability to reconstruct the initial density matrix $\rho(0)$.

We can formulate and prove the following theorem.
\begin{thm}\label{thm3.1}
The initial density matrix of a qutrit with evolution given by a phase-damping channel of the form \eqref{e3.1} can be uniquely determined on the basis of mean values of \textbf{two} Hermitian operators:
\begin{equation}\label{e3.5}
Q_1 =  \left [ \begin{matrix} 1 & 1  & -i \\ 1 & -1 & -i \\ i & i &0 \end{matrix} \right] \text{ and } Q_2 =  \left [ \begin{matrix} \frac{1}{\sqrt{3}} & -i  & 1 \\ i & \frac{1}{\sqrt{3}} & 1 \\ 1 & 1 & -\frac{2}{\sqrt{3}} \end{matrix} \right],
\end{equation}
which are measured at \textbf{four} time instants $t_1, t_2, t_3, t_4$ (the time instants have to be distinct, i.e.: $t_1 \neq t_2 \neq t_3 \neq t_4$).
\end{thm}

\begin{proof}
First, one has to notice that the dynamic matrix $D(t)$ (from \eqref{e3.3}) can be decomposed in the form:
\begin{equation}\label{e3.6}
D(t) = \mathbb{I}_3 + e^{-\gamma_1 t} A_1 + e^{-\gamma_2 t} A_2 + e^{-\gamma_3 t} A_3,
\end{equation}
where
\begin{equation}
A_1 = \left [ \begin{matrix} 0 & 1  & 0 \\ 1 & 0 & 0 \\ 0 & 0 & 0 \end{matrix} \right] \text{  } A_2 = \left [ \begin{matrix} 0 & 0 & 1 \\ 0 & 0 & 0 \\ 1 & 0 & 0 \end{matrix} \right] \text{ and }  A_3 = \left [ \begin{matrix} 0 & 0 & 0 \\ 0 & 0 & 1 \\ 0 & 1 & 0 \end{matrix} \right].
\end{equation}

Applying \eqref{21} into this example, one gets a formula for mean values of the observables:
\begin{equation}\label{e3.7}
\begin{aligned}
{}& m_i (t_j) = Tr\{ (Q_i \circ \mathbb{I}_3) \rho(0)\} + e^{-\gamma_1 t_j} Tr\{ (Q_i \circ A_1) \rho(0)\}+\\& + e^{-\gamma_2 t_j} Tr\{ (Q_i \circ A_2) \rho(0)\}+ e^{-\gamma_3 t_j} Tr\{ (Q_i \circ A_3) \rho(0)\}.
\end{aligned}
\end{equation}

In case of the observable $Q_1$ one can easily calculate that:
\begin{equation}
\begin{split}
& Q_1 \circ \mathbb{I}_3 = \lambda_3,\\
& Q_1 \circ A_1 = \lambda_1,\\
& Q_1 \circ A_2 = \lambda_5,\\
& Q_1 \circ A_3 = \lambda_7,
\end{split}
\end{equation}
which leads to the following formula for the measurement result:
\begin{equation}\label{e3.8}
\begin{aligned}
m_i (t_j) {}&= Tr\{ \lambda_3 \rho(0)\} + e^{-\gamma_1 t_j} Tr\{ \lambda_1 \rho(0)\} + \\& + e^{-\gamma_2 t_j} Tr\{ \lambda_5 \rho(0)\} + e^{-\gamma_3 t_j} Tr\{ \lambda_7 \rho(0)\}.
\end{aligned}
\end{equation}

Assuming that the mean value of the observable $Q_1$ can be obtained at $4$ distinct time instants $t_1, t_2, t_3, t_4$, one gets a matrix equation (compare with \eqref{22}):
\begin{equation}\label{e3.9}
\left[ \begin{matrix} m_1 (t_1) \\ m_1 (t_2) \\ m_1 (t_3) \\ m_1 (t_4) \end{matrix} \right] = 
\left[\begin{matrix} 1 & e^{-\gamma_1 t_1} & e^{-\gamma_2 t_1} & e^{-\gamma_3 t_1} \\ 1 & e^{-\gamma_1 t_2} & e^{-\gamma_2 t_2} & e^{-\gamma_3 t_2} \\ 1 & e^{-\gamma_1 t_3} & e^{-\gamma_2 t_3} & e^{-\gamma_3 t_3} \\ 1 & e^{-\gamma_1 t_4} & e^{-\gamma_2 t_4} & e^{-\gamma_3 t_4}
\end{matrix} \right] \left[ \begin{matrix} Tr\{ \lambda_3 \rho(0)\} \\ Tr\{ \lambda_1 \rho(0)\} \\ Tr\{ \lambda_5 \rho(0)\} \\  Tr\{ \lambda_7 \rho(0)\} \end{matrix} \right].
\end{equation}

In case of the observable $Q_2$ one can quickly calculate that:
\begin{equation}
\begin{split}
& Q_2 \circ \mathbb{I}_3 = \lambda_8,\\
& Q_2 \circ A_1 = \lambda_2,\\
& Q_2 \circ A_2 = \lambda_4,\\
& Q_2 \circ A_3 = \lambda_6,
\end{split}
\end{equation}
which yields the following formula for the mean value of $Q_2$ measured at the moment $t_j$:
\begin{equation}\label{e3.10}
\begin{aligned}
m_i (t_j) {}&= Tr\{ \lambda_8 \rho(0)\} + e^{-\gamma_1 t_j} Tr\{ \lambda_2 \rho(0)\}+\\& +e^{-\gamma_2 t_j} Tr\{ \lambda_4 \rho(0)\} + e^{-\gamma_3 t_j} Tr\{ \lambda_6 \rho(0)\}.
\end{aligned}
\end{equation}

Analogously like before, let us assume that the mean value of the observable $Q_2$ can be obtained at $4$ distinct time instants $t_1, t_2, t_3, t_4$. Then one gets a matrix equation (compare with \eqref{22}):
\begin{equation}\label{e3.11}
\left[ \begin{matrix} m_2 (t_1) \\ m_2 (t_2) \\ m_2 (t_3) \\ m_2 (t_4) \end{matrix} \right] = 
\left[\begin{matrix} 1 & e^{-\gamma_1 t_1} & e^{-\gamma_2 t_1} & e^{-\gamma_3 t_1} \\ 1 & e^{-\gamma_1 t_2} & e^{-\gamma_2 t_2} & e^{-\gamma_3 t_2} \\ 1 & e^{-\gamma_1 t_3} & e^{-\gamma_2 t_3} & e^{-\gamma_3 t_3} \\ 1 & e^{-\gamma_1 t_4} & e^{-\gamma_2 t_4} & e^{-\gamma_3 t_4}
\end{matrix} \right] \left[ \begin{matrix} Tr\{ \lambda_8 \rho(0)\} \\ Tr\{ \lambda_2 \rho(0)\} \\ Tr\{ \lambda_4 \rho(0)\} \\  Tr\{ \lambda_6 \rho(0)\} \end{matrix} \right].
\end{equation}

One can observe that the matrix equations \eqref{e3.9} and \eqref{e3.11} are uniquely solvable if and only if:
\begin{equation}\label{e3.12}
det \left[\begin{matrix} 1 & e^{-\gamma_1 t_1} & e^{-\gamma_2 t_1} & e^{-\gamma_3 t_1} \\ 1 & e^{-\gamma_1 t_2} & e^{-\gamma_2 t_2} & e^{-\gamma_3 t_2} \\ 1 & e^{-\gamma_1 t_3} & e^{-\gamma_2 t_3} & e^{-\gamma_3 t_3} \\ 1 & e^{-\gamma_1 t_4} & e^{-\gamma_2 t_4} & e^{-\gamma_3 t_4}
\end{matrix} \right] \neq 0
\end{equation}

In order to demonstrate that this condition is satisfied we can assume that the time instants $t_1, t_2, t_3, t_4$ are selected according to a certain rule. An assumption concerning the choice of time instants does not cause any loss in generality since one can agree that in an experiment it would be relatively easy to obey constraints imposed on the choice of moments of observation. If we assume that the time instants constitute an arithmetic sequence, defined as: $t_1=t$, $t_2=2t$, $t_3=3t$, $t_4=4t$. Then, the condition \eqref{e3.12} can be rewritten as:
\begin{equation}\label{determinant}
det \left[\begin{matrix} 1 & \xi_1 &  \xi_2 &  \xi_3 \\ 1 & (\xi_1)^2 &  (\xi_2)^2 &  (\xi_3)^2 \\ 1 & (\xi_1)^3 &  (\xi_2)^3 &  (\xi_3)^3 \\ 1 & (\xi_1)^4 &  (\xi_2)^4 &  (\xi_3)^4
\end{matrix} \right] \neq 0,
\end{equation}
where $\xi_i = e^{-\gamma_i t}$. One can observe that the rows in the matrix are linearly independent as long as $\gamma_1 \neq \gamma_2 \neq \gamma_3 \neq 0$, which means that the condition \eqref{determinant} is satisfied. This implies that by solving both matrix equations \eqref{e3.9} and \eqref{e3.11} one can obtain the mean values $Tr\{ \lambda_i \rho(0)\}$ where $i=1,\dots,8$.

The solutions of \eqref{e3.9} and \eqref{e3.11} can be obtained on the basis of the mean values of the two selected observables $Q_1$ and $Q_2$. Thus, one can conclude that the figures $Tr\{ \lambda_i \rho(0)\}$ for $i = 1,\dots, 8$ are computable from experimental data provided one knows the structure of the dynamic matrix $D(t)$.

If the solutions of \eqref{e3.9} and \eqref{e3.11} are substituted to \eqref{e3.4}, one gets the explicit formula for the initial density matrix $\rho(0)$, which means that the knowledge about the evolution of the system given by \eqref{e3.1} combined with the measurements of the two observables $Q_1$ and $Q_2$ leads to the ability to reconstruct the unknown density matrix. This conclusion finishes the proof.
\end{proof}

One should note that the theorem \ref{thm3.1} is existential and does not give a specific formula for the unknown density matrix. Due to the number of parameters, the analytical solution would have a long and unattractive form. Nevertheless, one can choose specific values of the decoherence parameters as well as specific time instants and then one is able to obtain a specific formula for $\rho(0)$ by following the general procedure introduced in this subsection. 

\section{Dynamic quantum tomography model for entangled qubits subject to pure decoherence}\label{entangled}

\subsection{Preliminaries}

Entangled qubits are a special case of quantum systems associated with Hilbert space such that $dim \mathcal{H} =4$. Thus, let us first analyze in general quantum tomography problem for $4-$level quantum systems subject to pure decoherence and then apply the results to entangled qubits as a specific example of such systems.

In order to decompose any density matrix of a $4-$level quantum system one can utilize the generalized Gell-Mann (GGM) basis (see Appendix \ref{appendix}), which consists of $15$ elements \cite{bertlmann08}. There are $6$ symmetric, $6$ antisymmetric and $3$ diagonal GGM matrices. One can expand any density matrix $\rho(0)$ according to the formula \eqref{general}.

In order to determine the initial density matrix of a $4-$level quantum system one would have to measure the mean value of every GGM operator, which would deliver the Bloch vector. However such an approach to quantum state tomography, although suggested in \cite{genki03}, does not seem practical. One is not able to define in a laboratory $15$ distinct measurable quantities.

\subsection{Evolution model for $4-$level systems}

Let us assume that the evolution of a $4-$level quantum system is given by a dynamical map of the form:

\begin{equation}\label{e4.1}
\rho(t) = \left[ \begin{matrix} 1 & e^{-\gamma_1 t} & e^{-\gamma_2 t} & e^{-\gamma_3 t} \\ e^{-\gamma_1 t} & 1 & e^{-\gamma_4 t} &e^{-\gamma_5 t} \\ e^{-\gamma_2 t} & e^{-\gamma_4 t} & 1 & e^{-\gamma_6 t} \\ e^{-\gamma_3 t} & e^{-\gamma_5 t} & e^{-\gamma_6 t} & 1 \end{matrix} \right] \circ \rho(0),
\end{equation}

where $\gamma_1, \dots, \gamma_6$ are positive decoherence rates. We assume that $\gamma_i \neq \gamma_j$ for $i\neq j$ and $\rho(0)\in \mathcal{S(H)}$ denotes the unknown initial density matrix.

One can quickly check that the dynamic matrix $D(t)$ from the definition \eqref{e4.1} satisfies the conditions \eqref{conditions}. Therefore, the channel defined in \eqref{e4.1} describes a legitimate dynamical map. Thus, we shall apply the reasoning from \cite{czerwin16b} in order to investigate how one can benefit from the knowledge about the evolution when solving quantum state tomography problem.

\subsection{Results and analysis}

The goal of this section is to demonstrate that the effectiveness of quantum state tomography can be significantly improved if we assume that the evolution of the quantum system is given by a phase-damping channel. We can formulate and prove a theorem.

\begin{thm}\label{thm4.1}
The initial density matrix $\rho(0)$ of a $4-$level quantum system with evolution given by a phase-damping channel of the form \eqref{e4.1} can be uniquely determined on the basis of mean values of \textbf{two} observables:
\begin{equation}\label{e4.2}
Q_1 =  \left [ \begin{matrix} 1 & 1  & -i & 1 \\ 1 & -1 & -i & 1 \\ i & i & 0 & -i \\ 1 & 1 & i &0 \end{matrix} \right] \text{ and } Q_2 =  \left [ \begin{matrix} \frac{1}{\sqrt{3}} & -i  & 1 & -i \\ i & \frac{1}{\sqrt{3}} & 1 & -i \\ 1 & 1 & -\frac{2}{\sqrt{3}} & 1 \\ i & i & 1& 0 \end{matrix} \right],
\end{equation}
which are measured at \textbf{seven} distinct time instants,\\
and the mean value of the operator:
\begin{equation}\label{e4.3}
\Lambda_3 = \frac{1}{\sqrt{6}}  \left [ \begin{matrix} 1 & 0  & 0 & 0 \\ 0 & 1 & 0 & 0 \\ 0 & 0 & 1 & 0 \\ 0 & 0 & 0 & -3 \end{matrix} \right]
\end{equation}
measured once at the time instant $t=0$.
\end{thm}

\begin{proof}
First, one has to notice that the dynamic matrix $D(t)$ (from the channel definition \eqref{e4.1}) can be decomposed in the form:
\begin{equation}\label{e4.4}
\begin{aligned}
D(t) {}& = \mathbb{I}_4 + e^{-\gamma_1 t} A_1 + e^{-\gamma_2 t} A_2 + e^{-\gamma_3 t} A_3 + e^{-\gamma_4 t} A_4 +\\
& + e^{-\gamma_5 t} A_5 + e^{-\gamma_6 t} A_6,
\end{aligned}
\end{equation}
where
\begin{equation}\label{e4.5}
\begin{split}
& A_1 = \left [ \begin{matrix} 0 & 1  & 0 & 0 \\ 1 & 0 & 0 &0 \\ 0 & 0 & 0 & 0 \\ 0 & 0 & 0 & 0 \end{matrix} \right] \hspace{0.1cm} A_2 = \left [ \begin{matrix} 0 & 0 & 1 & 0 \\ 0 & 0 & 0 & 0 \\ 1 & 0 & 0 & 0  \\ 0 & 0 & 0 & 0 \end{matrix} \right] \hspace{0.1cm}  A_3 = \left [ \begin{matrix} 0 & 0 & 0 &1  \\ 0 & 0 & 0 & 0 \\ 0 & 0 & 0 & 0  \\ 1 & 0 & 0 &0 \end{matrix} \right]\\
& A_4 = \left [ \begin{matrix} 0 & 0 & 0 & 0 \\ 0 & 0 & 1 & 0 \\0 & 1 & 0 & 0  \\ 0 & 0 & 0 & 0 \end{matrix} \right] \hspace{0.1cm} A_5  = \left [ \begin{matrix} 0 & 0 & 0 &0  \\ 0 & 0 & 0 & 1 \\ 0 & 0 & 0 & 0  \\ 0 & 1 & 0 &0 \end{matrix} \right] \hspace{0.1cm} A_6 = \left [ \begin{matrix} 0 & 0 & 0 & 0 \\ 0 & 0 & 0 & 0 \\0 & 0 & 0 & 1  \\ 0 & 0 & 1 & 0 \end{matrix} \right].
\end{split}
\end{equation}

Applying this decomposition into \eqref{21}, one gets a formula for mean values of the observables:
\begin{equation}\label{e4.6}
\begin{aligned}
{}& m_i (t_j)= Tr\{ (Q_i \circ \mathbb{I}_4) \rho(0)\} + e^{-\gamma_1 t_j} Tr\{ (Q_i \circ A_1) \rho(0)\}+\\& + e^{-\gamma_2 t_j} Tr\{ (Q_i \circ A_2) \rho(0)\}+ e^{-\gamma_3 t_j} Tr\{ (Q_i \circ A_3) \rho(0)\}+\\& + e^{-\gamma_4 t_j} Tr\{ (Q_i \circ A_4) \rho(0)\}+ e^{-\gamma_5 t_j} Tr\{ (Q_i \circ A_5) \rho(0)\}+\\& + e^{-\gamma_6 t_j} Tr\{ (Q_i \circ A_6) \rho(0)\}.
\end{aligned}
\end{equation}

In case of the observable $Q_1$ one can notice that:
\begin{equation}\label{e4.7}
\begin{split}
& Q_1 \circ \mathbb{I}_4 = \Lambda_1, \hspace{0.75cm} Q_1 \circ A_1 = \Lambda_s^{12},\\
& Q_1 \circ A_2 = \Lambda_a^{13}, \hspace{0.5cm} Q_1 \circ A_3 = \Lambda_s^{14}\\
& Q_1 \circ A_4 = \Lambda_a^{23}, \hspace{0.5cm} Q_1 \circ A_5 = \Lambda_s^{24},\\
& Q_1 \circ A_6 = \Lambda_a^{34},
\end{split}
\end{equation}
which leads to the following formula for the measurement result:
\begin{equation}\label{e4.8}
\begin{aligned}
 m_i (t_j) {}& = Tr\{ \Lambda_1 \rho(0)\} + e^{-\gamma_1 t_j} Tr\{ \Lambda_s^{12} \rho(0)\} + \\& + e^{-\gamma_2 t_j} Tr\{ \Lambda_a^{13} \rho(0)\} + e^{-\gamma_3 t_j} Tr\{ \Lambda_s^{14} \rho(0)\} + \\& +e^{-\gamma_4 t_j} Tr\{ \Lambda_a^{23} \rho(0)\} + e^{-\gamma_5 t_j} Tr\{ \Lambda_s^{24} \rho(0)\}+\\&+ e^{-\gamma_6 t_j} Tr\{ \Lambda_a^{34} \rho(0)\}.
\end{aligned}
\end{equation}

Assuming that the mean value of the observable $Q_1$ can be obtained at $7$ distinct time instants $t_1, t_2, t_3, t_4,t_5,t_6,t_7$, one gets a matrix equation (analogous to \eqref{e3.11}):
\begin{equation}\label{e4.9}
\left[ \begin{matrix} m_1 (t_1) \\ m_1 (t_2) \\ \vdots \\ m_1 (t_7)  \end{matrix} \right] = 
\left[\begin{matrix} 1 & e^{-\gamma_1 t_1} & \cdots & e^{-\gamma_6 t_1} \\ 1 & e^{-\gamma_1 t_2} & \cdots & e^{-\gamma_6 t_2} \\ \vdots & \vdots & \ddots & \vdots \\ 1 & e^{-\gamma_1 t_7} & \cdots & e^{-\gamma_6 t_7}
\end{matrix} \right] \left[ \begin{matrix} Tr\{ \Lambda_1 \rho(0)\} \\ Tr\{ \Lambda_s^{12} \rho(0)\} \\ \vdots \\  Tr\{ \Lambda_a^{34} \rho(0)\} \end{matrix} \right].
\end{equation}

On the other hand, for $Q_2$ one can quickly obtain:
\begin{equation}\label{e4.10}
\begin{split}
& Q_2 \circ \mathbb{I}_4 = \Lambda_2, \hspace{0.75cm} Q_2 \circ A_1 = \Lambda_a^{12},\\
& Q_2 \circ A_2 = \Lambda_s^{13}, \hspace{0.5cm} Q_2 \circ A_3 = \Lambda_a^{14}\\
& Q_2 \circ A_4 = \Lambda_s^{23}, \hspace{0.5cm} Q_2 \circ A_5 = \Lambda_a^{24},\\
& Q_2 \circ A_6 = \Lambda_s^{34},
\end{split}
\end{equation}
which gives us the following equation for the mean value of $Q_2$ at the moment $t_j$:
\begin{equation}\label{e4.11}
\begin{aligned}
 m_i (t_j) {}& = Tr\{ \Lambda_2 \rho(0)\} + e^{-\gamma_1 t_j} Tr\{ \Lambda_a^{12} \rho(0)\} + \\& + e^{-\gamma_2 t_j} Tr\{ \Lambda_s^{13} \rho(0)\} + e^{-\gamma_3 t_j} Tr\{ \Lambda_a^{14} \rho(0)\} + \\& +e^{-\gamma_4 t_j} Tr\{ \Lambda_s^{23} \rho(0)\} + e^{-\gamma_5 t_j} Tr\{ \Lambda_a^{24} \rho(0)\}+\\&+ e^{-\gamma_6 t_j} Tr\{ \Lambda_s^{34} \rho(0)\}.
\end{aligned}
\end{equation}

We operate in the same vein as for the observable $Q_1$ and after $7$ measurements at distinct time moments, we get a matrix equation:
\begin{equation}\label{e4.12}
\left[ \begin{matrix} m_2 (t_1) \\ m_2 (t_2) \\ \vdots \\ m_2 (t_7)  \end{matrix} \right] = 
\left[\begin{matrix} 1 & e^{-\gamma_1 t_1} & \cdots & e^{-\gamma_6 t_1} \\ 1 & e^{-\gamma_1 t_2} & \cdots & e^{-\gamma_6 t_2} \\ \vdots & \vdots & \ddots & \vdots \\ 1 & e^{-\gamma_1 t_7} & \cdots & e^{-\gamma_6 t_7} \end{matrix} \right] \left[ \begin{matrix} Tr\{ \Lambda_2 \rho(0)\} \\ Tr\{ \Lambda_a^{12} \rho(0)\} \\ \vdots \\  Tr\{ \Lambda_s^{34} \rho(0)\} \end{matrix} \right].
\end{equation}

One can observe that the matrix equations \eqref{e4.9} and \eqref{e4.12} are uniquely solvable if and only if:
\begin{equation}\label{e4.13}
det \left[\begin{matrix} 1 & e^{-\gamma_1 t_1} & \cdots & e^{-\gamma_6 t_1} \\ 1 & e^{-\gamma_1 t_2} & \cdots & e^{-\gamma_6 t_2} \\ \vdots & \vdots & \ddots & \vdots \\ 1 & e^{-\gamma_1 t_7} & \cdots & e^{-\gamma_6 t_7} \end{matrix} \right] \neq 0
\end{equation}

In order to demonstrate that the condition \eqref{e4.13} is satisfied we follow the same strategy as in case of qutrits. We assume that the time instants are not arbitrary and they constitute an arithmetic sequence defined as: $t_1=t$, $t_2=2t$, $\dots$, $t_7=7t$. Then, the condition \eqref{e4.13} can be rewritten as:
\begin{equation}\label{determinant2}
det \left[\begin{matrix} 1 & \xi_1 &  \cdots &  \xi_6 \\ 1 & (\xi_1)^2 & \cdots &  (\xi_6)^2 \\ \vdots & \vdots & \ddots & \vdots\\ 1 & (\xi_1)^7 &  \cdots &  (\xi_6)^7
\end{matrix} \right] \neq 0,
\end{equation}
where $\xi_i = e^{-\gamma_i t}$ (for $i=1,\dots, 6$). One can observe that the rows in the matrix are linearly independent as long as the decoherence rates are positive and satisfy $\gamma_i \neq \gamma_j \neq 0$ (for $i\neq j$). This means that the condition \eqref{determinant2} holds true. Consequently, by solving both matrix equations \eqref{e4.9} and \eqref{e4.12} one obtains the mean values of $14$ operators which belong to the $SU(4)$ generators. In general, a density matrix of a $4-$level quantum system is characterized by $15$ parameters. Thus, one needs to perform the measurement of $\Lambda_3$ at time instant $t=0$ in order to get a complete set of information. When one knows the mean values of the GGM operators, one is able to reconstruct the initial density matrix, which finishes the proof.
\end{proof}

\subsection{Special case: entangled qubits}

In general, the density matrix of a $4-$level quantum system is fully characterized by $15$ real parameters. However, if one has a priori knowledge about the system in question, one can simplify the algebraic structure of the density matrix. In case of $dim \mathcal{H} =4$ one may consider entangled qubits as a specific example of quantum system.

Mixed-state entanglement can be defined by means of statistical mixture of the Bell states \cite{bennet96}. We shall introduce the following form of the initial quantum state of entangled qubits:
\begin{equation}\label{e4.14}
\begin{aligned}
\rho (0) ={}& p_1 \ket{\Phi^+}\bra{\Phi^+} + p_2 \ket{\Phi^-}\bra{\Phi^-} + p_3 \ket{\Psi^+}\bra{\Psi^+} +\\
&+ (1-(p_1+p_2+p_3))\ket{\Psi^-}\bra{\Psi^-}, 
\end{aligned}
\end{equation}
where the state vectors $\{\ket{\Phi^+},  \ket{\Phi^-}, \ket{\Psi^+}, \ket{\Psi^-} \}$ denote the Bell basis in the Hilbert space and $p_1, p_2, p_3$ are probabilities.

The density matrix defined in \eqref{e4.14} covers a wide range of particular mixed entangled states, for example the famous Werner state \cite{werner89}.

In order to investigate the problem of quantum tomography, let us first observe that our state $\rho (0)$ can be decomposed in the basis of $SU(4)$ generators:
\begin{equation}\label{e4.15}
\begin{aligned}
\rho = {}& \frac{1}{4} \mathbb{I}_4 + \frac{p_1 - p_2}{2} \Lambda_s ^{14} + \frac{1-p_1-p_2 -2p_3}{2} \Lambda_s^{23} +\\
&+ \frac{2p_1+2p_2-1}{4} \Lambda_1 + \frac{2p_1 + 2p_2 -1}{4 \sqrt{3}} \Lambda_2 +\\
& + \frac{1-2p_1 - 2p_2}{2\sqrt{6}} \Lambda_3.
\end{aligned}
\end{equation}

One is intuitively aware that we need three independent pieces of information in order to reconstruct the density matrix. Thus, one could measure three operators from the set of the GGM matrices: $\Lambda_s ^{14},  \Lambda_s^{23}, \Lambda_1$ at time instant $t=0$, which would give the probabilities $p_1, p_2, p_3$ required to describe the density matrix. However, we propose a quantum tomography scheme based on dynamic approach which leads to a significant improvement. Let us prove a theorem.

\begin{thm}
Assuming that the evolution of $\rho(0)$ is given by the phase-damping channel defined in \eqref{e4.1}, one can uniquely determine the probabilities $p_1, p_2, p_3$ (which completely characterize the entagled qubits state) by performing the measurement of one observable $Q$ at three distinct time instants, where
\begin{equation}\label{e4.16}
Q =  \left [ \begin{matrix} 1 & 0  & 0 & 1 \\0 & -1 & 1 & 0 \\ 0 & 1 & 0 & 0 \\ 1 & 0 & 0 &0 \end{matrix} \right].
\end{equation}
\end{thm}

\begin{proof}
We assume that the evolution of mixed-entangled qubits is subject to the same dynamical map as introduced in \eqref{e4.1}. Therefore, the decompositions given in the equations \eqref{e4.4}-\ref{e4.6} still hold true.

In case of the observable $Q$ one can immediately notice that
\begin{equation}\label{e4.17}
\begin{split}
& Q \circ \mathbb{I}_4 = \Lambda_1, \hspace{0.75cm} Q \circ A_1 = 0,\\
& Q \circ A_2 = 0, \hspace{0.5cm} Q \circ A_3 = \Lambda_s^{14}\\
& Q \circ A_4 = \Lambda_s^{23}, \hspace{0.5cm} Q \circ A_5 = 0,\\
& Q \circ A_6 = 0,
\end{split}
\end{equation}
where $0$ denotes here the $4\times4$ dimensional zero matrix. Then, the formula for the mean value of the observable $Q$ measured at any time instant $t_j$ can be simplified:
\begin{equation}\label{e4.18}
\begin{aligned}
 m (t_j)  = {}& Tr\{ \Lambda_1 \rho(0)\} + e^{-\gamma_3 t_j} Tr\{ \Lambda_s^{14} \rho(0)\} + \\
& + e^{-\gamma_4 t_j} Tr\{ \Lambda_s^{23} \rho(0)\}.
\end{aligned}
\end{equation}

Assuming that the mean value of $Q$ can be obtained at three distinct time instants $t_1, t_2, t_3$, one gets a matrix equation:
\begin{equation}\label{e4.19}
\left[ \begin{matrix} m (t_1) \\ m (t_2) \\ m (t_3) \end{matrix} \right] = 
\left[\begin{matrix} 1 & e^{-\gamma_3 t_1} & e^{-\gamma_4 t_1} \\ 1 & e^{-\gamma_3 t_2} & e^{-\gamma_4 t_2} \\ 1 & e^{-\gamma_3 t_3} & e^{-\gamma_4 t_3} 
\end{matrix} \right] \left[ \begin{matrix} Tr\{ \Lambda_1 \rho(0)\} \\ Tr\{\Lambda_s^{14} \rho(0)\} \\ Tr\{ \Lambda_s^{23} \rho(0)\} \} \end{matrix} \right].
\end{equation}

One can compute the expectation values: $\{Tr\{ \Lambda_1 \rho(0)\}, Tr\{\Lambda_s^{14} \rho(0)\}, Tr\{ \Lambda_s^{23} \rho(0)\}\}$ on the basis of the experimental data $\{m (t_1), m (t_2), m (t_3) \}$ if and only if:
\begin{equation}\label{e4.20}
\left[\begin{matrix} 1 & e^{-\gamma_3 t_1} & e^{-\gamma_4 t_1} \\ 1 & e^{-\gamma_3 t_2} & e^{-\gamma_4 t_2} \\ 1 & e^{-\gamma_3 t_3} & e^{-\gamma_4 t_3} 
\end{matrix} \right] \neq 0
\end{equation}

Likewise before, we may assume that the time instants are selected in such a way that they are elements of an arithmetic sequence: $t_1=t$, $t_2 = 2 t$, $t_3  = 3t$. Consequently, the condition \eqref{e4.20} can be reformulated:
\begin{equation}\label{determinant3}
det \left[\begin{matrix} 1 & \xi_3 &  \xi_4  \\ 1 & (\xi_3)^2 &  (\xi_4)^2 \\ 1 & (\xi_3)^3 &  (\xi_4)^3 \end{matrix} \right] \neq 0,
\end{equation}
where $\xi_i = e^{-\gamma_i t}$. One can observe that the rows are linearly independent as long as $\gamma_3 \neq \gamma_4$. Other decoherence rates does not affect the state of the entangled qubits and for this reason they might be even zeros.

Assuming that $\gamma_3 \neq \gamma_4$, one can calculate the set of mean values $\{Tr\{ \Lambda_1 \rho(0)\}, Tr\{\Lambda_s^{14} \rho(0)\}, Tr\{ \Lambda_s^{23} \rho(0)\}\}$  which suffice to determine the probabilities $p_1, p_2, p_3$ characterizing the initial density matrix.
\end{proof}

Entangled qubits state defined in \eqref{e4.14}, which is completely characterized by three probabilities $\{p_1, p_2, p_3\}$, can be reconstructed from the mean values of one observable measured at three different time instants.

The quantum tomography scheme for entangled qubits, which requires only one kind of measurement to reconstruct the initial density matrix, fits well to the current trends in the field. In the context of the dynamic approach to tomography, it has been discussed that for some specific evolution models one observable is sufficient to determine the initial density matrix \cite{czerwin16a}. This concept is not only theoretical because experimental realizations have demonstrated recently that a single observable (i.e. one measurement setup) provides all necessary data for quantum state tomography \cite{oren17}. In quantum process tomography there is a similar tendency to develop methods which enable to reconstruct unknown quantum channels from few measurements \cite{kliesch19}.

\section{Quantum tomography of qudits subject to pure decoherence}

\subsection{Preliminaries}

For a qudit one can always use the generalized Gell-Mann (GGM) basis to decompose any density matrix by means of $N^2 - 1$ operators (see details in Appendix \ref{appendix}). The coefficients which appear in the expansion \eqref{general} can be interpreted as mean values of the GGM operators (they constitute the Bloch vector). Thus, one needs to perform $N^2 - 1$ distinct measurements in order to decompose the initial density matrix. This approach seems impractical, especially for higher dimensions since the number of required observables increases quadratically.

\subsection{Evolution model for qudits}

Time evolution of qudit is given by a phase-damping channel according to \eqref{17}, where the dynamic matrix $D(t)$ is symmetric and shall be expressed as:
\begin{equation}\label{e4.21}
D(t) = \left[ \begin{matrix} 1 & e^{-\gamma_{12} t}& e^{-\gamma_{13} t} & \cdots & e^{-\gamma_{1N} t} \\ e^{-\gamma_{12} t} & 1 & e^{-\gamma_{23} t} &\cdots & e^{-\gamma_{2N} t} \\ e^{-\gamma_{13} t} & e^{-\gamma_{23} t} & 1 & \dots &e^{-\gamma_{3N} t}\\ \vdots & \vdots & \vdots & \ddots & \vdots \\  e^{-\gamma_{1N} t} & e^{-\gamma_{2N} t} & e^{-\gamma_{3N} t} &\cdots & 1 \end{matrix} \right],
\end{equation}
where we assume that no two decoherence rates are the same.

\subsection{Results and analysis}

We shall investigate what can be said about the benefits of the dynamic approach to quantum tomography of systems subject to pure decoherence. We formulate two theorems -- one tells about the advantages of the dynamic tomography, whereas the other indicates a serious limitation of this method. The first theorem demonstrates that one can in general very efficiently obtain the off-diagonal elements of the density matrix. On the other hand, the latter describes the worst-case scenario which is connected with computing the diagonal part of the unknown matrix. The theorem gives the formula for the upper limit of the number of distinct observables required for quantum tomography.

\begin{thm}
For a qudit subject to pure decoherence with the dynamic matrix defined in \eqref{e4.21} one can reconstruct all off-diagonal elements of the unknown density matrix $\rho(0)$ from the mean values of two observables measured at $\frac{N(N-1)}{2}+1$ distinct time instants.
\end{thm}

\begin{proof}
Let us start by representing $D(t)$ as:
\begin{equation}\label{e4.22}
D(t) = \mathbb{I}_N + \sum_{k=j+1}^N \sum_{j=1}^{N-1} e^{-\gamma_{jk}t} A_{jk},
\end{equation}
where $A_{jk} = E_{jk} + E_{kj}$ and $E_{jk}$ denotes a matrix from the standard basis.

Now we define $2$ observables by means of the GGM operators (see Appendix \ref{appendix}):
\begin{equation}\label{e4.23}
\begin{split}
& Q_1 = \Lambda_1 + \sum_{k= 2i}^{N} \sum_{i=1}^{\kappa_1(N)} \Lambda_s ^{(2i-1)k} +\sum_{k= 2i+1}^{N} \sum_{i=1}^{\kappa_2(N)} \Lambda_a ^{(2i)k},\\
& Q_2 = \Lambda_2 + \sum_{k= 2i+1}^{N} \sum_{i=1}^{\kappa_2(N)} \Lambda_s ^{(2i)k} + \sum_{k= 2i}^{N} \sum_{i=1}^{\kappa_1(N)} \Lambda_a ^{(2i-1)k},
\end{split}
\end{equation}
where
\begin{equation}\label{e4.24}
\begin{split}
& \kappa_1 (N) = \begin{cases} \frac{N-1}{2} &\mbox{  when  } N \mbox{ is odd}  \\ \frac{N}{2} &\mbox{  when  } N \mbox{ is even} \end{cases}\\
&\kappa_2 (N) = \begin{cases} \frac{N-1}{2} &\mbox{  when  } N \mbox{ is odd}  \\ \frac{N-2}{2} &\mbox{  when  } N \mbox{ is even} \end{cases}.
\end{split}
\end{equation}

Then for any time instant $t_j$ the mean values of $Q_1$ and $Q_2$ can be written as:
\begin{equation}\label{e4.25}
\begin{aligned}
m_1 (t_j) = &{} Tr \{\Lambda_1 \rho(0)\} + \sum_{k= 2i}^{N} \sum_{i=1}^{\kappa_1(N)} e^{-\gamma_{(2i-1)k}t} Tr\{\Lambda_s ^{(2i-1)k} \rho(0)\}  \\&+ \sum_{k= 2i+1}^{N} \sum_{i=1}^{\kappa_2(N)} e^{-\gamma_{(2i)k}t} Tr\{\Lambda_a ^{(2i)k} \rho(0)\}
\end{aligned}
\end{equation}
\begin{equation}\label{e4.26}
\begin{aligned}
m_2 (t_j) = &{} Tr \{\Lambda_2 \rho(0)\} + \sum_{k= 2i+1}^{N} \sum_{i=1}^{\kappa_2(N)} e^{-\gamma_{(2i)k}t} Tr\{\Lambda_s ^{(2i)k} \rho(0)\}  \\&+ \sum_{k= 2i}^{N} \sum_{i=1}^{\kappa_1(N)} e^{-\gamma_{(2i-1)k}t} Tr\{\Lambda_a ^{(2i-1)k} \rho(0)\}
\end{aligned}
\end{equation}

One can notice that if we measure both observables at $\frac{N(N-1)}{2}+1$ time instants selected in such a way that $t_1=t$, $t_2 = 2t$, $t_3 = 3 t, \dots$ then from \eqref{e4.25}-\ref{e4.26} we can compute the mean values of $N(N-1)+2$ operators: i.e. the $\frac{N(N-1)}{2}$ symmetric GGM, the $\frac{N(N-1)}{2}$ antisymmetric GGM and the two diagonal GGM $\Lambda_1$ and $\Lambda_2$. This means that the measurements provide $N(N-1)+2$ independent pieces of information which are sufficient to determine the off-diagonal elements of the density matrix.
\end{proof}

\begin{thm}
For a qudit subject to pure decoherence with evolution governed by a dynamic matrix defined in \eqref{e4.21} the upper boundary of the number of distinct observables required for state reconstruction equals: $N-1$.
\end{thm}

\begin{proof}
We obtain complete knowledge about the off-diagonal elements of the density matrix on the basis of mean values of the operators from \eqref{e4.23}. However the dynamic approach does not give any advantage when it comes to the diagonal elements since a phase-damping channel does not affect the diagonal of $\rho(0)$). Therefore, in the worst-case scenario, one would additionally need to measure the mean values of the remaining $N-3$ diagonal GGM matrices at time instant $t=0$, which means that finally one would use $N-1$ distinct observables to reconstruct the density matrix. The final formula for the initial state could be written according to \eqref{general}.
\end{proof}

For specific qudits the minimal number of distinct observables might be lower than $N-1$ since one may a priori know that there are some zeros on the diagonal of $\rho(0)$.

\section{Conclusions}

In this article we have proved that the dynamic approach to quantum tomography can be an efficient method of state reconstruction for qutrits and entangled qubits. It was demonstrated that one can significantly decrease the number of distinct observables required for state reconstruction if the system is sent through a phase-damping quantum channel. The results are in line with recent developments in quantum tomography where there is a tendency to search for economic methods which aim to reduce the number of distinct measurement setups \cite{gross10,shabani11,huszar12}.  

Furthermore, we formulated two general theorems concerning the criteria for quantum tomography of qudits subject to pure decoherence. The first theorem describes optimal benefits that one can obtain from applying this approach, i.e. all off-diagonal elements can be computed from the mean values of two observables. Whereas the other theorem expresses a key limitation of the method. If an $N-$level open quantum system is subject to pure decoherence, in the worst-case scenario, one needs $N-1$ distinct measurements to reconstruct the initial state.

\section*{Acknowledgments}

The author acknowledges financial support from the Foundation for Polish Science (FNP) (project First Team co-financed by the European Union under the European Regional Development Fund).

\appendix
\section{Generalized Gell-Mann matrices}\label{appendix}
In general, all GGM matrices can be divided into $3$ groups and for $dim \mathcal{H} =N$ they are defined as \cite{genki03,bertlmann08}:

i) $\frac{N(N-1)}{2}$ symmetric GGM matrices
\begin{equation}\label{symmetric}
\Lambda_s ^{jk} = \ket{j}\bra{k} + \ket{k}\bra{j}, \hspace{0.5cm} 1 \leq j < k \leq N,
\end{equation}

ii) $\frac{N(N-1)}{2}$ antisymmetric GGM matrices
\begin{equation}\label{antisymmetric}
\Lambda_a ^{jk} = -i \ket{j}\bra{k} + i \ket{k}\bra{j}, \hspace{0.5cm} 1 \leq j < k \leq N,
\end{equation}

iii) $(N-1)$ diagonal GGM matrices
\begin{equation}\label{diagonal}
\begin{split}
&\Lambda^l = \sqrt{\frac{2}{l(l+1)}} \left( \sum_{j=1}^l \ket{j}\bra{j} - l\ket{l+1}\bra{l+1} \right),\\
&1 \leq l \leq N-1,
\end{split}
\end{equation}
where $\ket{k}, \ket{j}, \ket{l}$ denote vectors from the standard basis.

One can notice that in total we have $N^2-1$ GGM matrices. From the definitions we can see that all GGM matrices are traceless and self-adjoint. It can be proved that they are orthogonal and form a basis. Thus, they are the generators of $SU(N)$.

A density matrix associated with a Hilbert space such that $dim \mathcal{H} =N$ can be decomposed according to \cite{genki03}:
\begin{equation}\label{general}
\rho = \frac{1}{N} \mathbb{I}_N + \frac{1}{2} \vec{s} \cdot \hat{\Lambda},
\end{equation}
where $\vec{s}$ denotes the Bloch vector and $\hat{\Lambda}$ the vector generated from ordered GGM matrices. The elements of the Bloch vector are equal to the mean values of the corresponding GGM operators.

\end{document}